\documentclass[pra,aps,twocolumn,superscriptaddress]{revtex4-1}
\usepackage{amsmath,amssymb,amsthm,dsfont,hyperref,graphicx}

\newcommand{\bra}[1]{\langle#1|}
\newcommand{\ket}[1]{|#1\rangle}

\newcommand{\Tr}{\operatorname{Tr}}

\newcommand{\hilb}[1]{\mathcal{#1}}

\newcommand{\map}[1]{\mathcal{#1}}

\newcommand{\linear}[1]{\boldsymbol{\mathsf{L}}(#1)}
\newcommand{\states}[1]{\boldsymbol{\mathsf{S}}(#1)}
\newcommand{\id}{\mathsf{id}}
\newcommand{\choi}[1]{\boldsymbol{\mathsf{C}}\!\left[{#1}
    \right]}
\newcommand{\invchoi}[1]{\boldsymbol{\mathsf{C}}^{-1}\!\left[{#1}
    \right]}
\renewcommand{\ge}{\geqslant}
\renewcommand{\le}{\leqslant}
\renewcommand{\i}{\operatorname{\imath}}
\def\sH{{\hilb{H}}}
\def\sK{{\hilb{K}}}
\def\>{\rangle}
\def\<{\langle}
\def\mL{{\map{L}}}
\def\mE{{\map{E}}}
\def\mT{{\map{T}}}
\def\openone{\mathds{1}}

\newtheorem{lmm}{Lemma}
\newtheorem{thm}{Proposition}

\begin{document}

\title{Direct observation of any two-point quantum correlation function}

\author{Francesco \surname{Buscemi}}

\email{buscemi@iar.nagoya-u.ac.jp}
\affiliation{Institute for Advanced Research, Nagoya
  University, Chikusa-ku, Nagoya 464-8601, Japan}

\affiliation{Graduate School of Information Science, Nagoya
  University, Chikusa-ku, Nagoya, 464-8601, Japan}

\author{Michele \surname{Dall'Arno}}

\email{michele.dallarno@math.cm.is.nagoya-u.ac.jp}
\affiliation{Graduate School of Information Science, Nagoya
  University, Chikusa-ku, Nagoya, 464-8601, Japan}

\author{Masanao \surname{Ozawa}}

\email{ozawa@is.nagoya-u.ac.jp}
\affiliation{Graduate School of Information Science, Nagoya
  University, Chikusa-ku, Nagoya, 464-8601, Japan}

\author{Vlatko \surname{Vedral}}

\email{phyvv@nus.edu.sg}
\affiliation{Atomic and Laser Physics, Clarendon Laboratory,
  University of Oxford, Parks Road, Oxford OX13PU, United
  Kingdom}

\affiliation{Center for Quantum Technologies, National
  University of Singapore, Republic of Singapore}

\begin{abstract}
  The existence of noncompatible observables in quantum
  theory makes a direct operational interpretation of
  two-point correlation functions problematic. Here we
  challenge such a view by explicitly constructing a
  measuring scheme that, independently of the input state
  $\rho$ and observables $A$ and $B$, performs an unbiased
  optimal estimation of the two-point correlation function
  $\Tr[A \ \rho \ B]$. This shows that, also in quantum
  theory, two-point correlation functions are as operational
  as any other expectation value. A very simple probabilistic
  implementation of our proposal is presented.
\end{abstract}

\maketitle

In the description of stochastic physical processes, it is
important to determine how dynamical variables are
correlated with each other. Correlation functions are
computed as products of two (or more) dynamical variables
(or their powers), averaged over time, or over many sites,
or in both ways. In the simplest case of the average of the
product of two dynamical variables, one usually speaks of
\emph{two-point correlation functions}.

In classical physics, dynamical variables are real-valued
functions of the state of the system. In fact, the state of
the system can be fully specified, in principle, by giving
the values of all its dynamical variables (or its generating
set of variables), at any instant in time. In classical
statistical mechanics, therefore, there is no difficulty in
defining and computing correlation functions, however
complicated, between dynamical variables; as dynamical
variables are all experimentally accessible, so are all
correlation functions.

In quantum mechanics, on the contrary, the relation between
states and dynamical variables is much more subtle. In
particular, the notion of dynamical variables is replaced by
that of \emph{observables}, namely, self-adjoint operators
that can or cannot commute; this is the formal reason for
the existence of ``incompatible'' variables that cannot
simultaneously assume definite values in any
state~\cite{uncertanty}. This feature arguably lies at the
origin of all ``quantum spooks,'' including a prevailing
view that correlation functions are typically ill-defined
for a quantum mechanical system --- if two observables do
not both have a definite value simultaneously, how could one
compute the average of their product then?

Interpretational problems notwithstanding, one still can
formally define \emph{two-point quantum correlation
	functions} as $\Tr[A\ \rho\ B]$, where $\rho$ describes
the state of the system and $A$, $B$ are any two observables
(or, possibly, the same observable at different times, in
which case we more precisely speak of \emph{auto-correlation
	functions}). In fact, such functions are extensively used
in a wide variety of fields, such as quantum statistical
mechanics~\cite{quant-stat}, quantum
thermodynamics~\cite{quant-term}, and quantum field
theory~\cite{quant-field}. The question then naturally
arises, whether quantum correlation functions can be given a
clear operational interpretation.

For the sake of concreteness, suppose we are given a source,
in control, emitting independent systems, all of them in the
same (though unknown) state. We are also given two measuring
apparatuses, as accurate as the theory (classical or
quantum) allows, one to measure observable $A$, the other
for observable $B$. We assume that both apparatuses can be
initialized and re-used an arbitrary number of times. In
classical physics, these assumptions are enough to allow us
to measure, in principle, not only the expected values of
$A$ and $B$, but also any moment of these, i.e. any two-point correlation function. This is
possible because, classically, measurement does not imply
disturbance. Therefore, one can perform successive
measurements of $A$ and $B$ on the same system, collect the
results, and post-process them at will. For example,
auto-correlation functions in classical physics are computed
by measuring a certain observable twice on the same system,
at different times. On the other hand, in quantum theory,
such a simple approach is often impossible due to the
existence of incompatible observables, as a measurement done
now unavoidably disturbs the evolution of the system and,
therefore, the result of a measurement performed on the
system at later times, unless the measurement satisfies
quantum non-demolition conditions; see for instance
Ref.~\cite{CTDSZ80}.

For this reason it seems that two-point correlation functions (and auto-correlation functions, in particular) cannot be interpreted operationally in quantum theory, in the sense that they cannot be directly measured experimentally. In this paper we argue that this would be too hurried a conclusion. Our contribution is to construct a
``black-box''--like approach to quantum correlation functions, working for (but being
independent of) any state $\rho$ and any pair of observables $A$ and
$B$ (see Fig.~\ref{fig:black-box}
below).

\begin{figure}[htb]
  \begin{center}
    \includegraphics[width=\columnwidth]{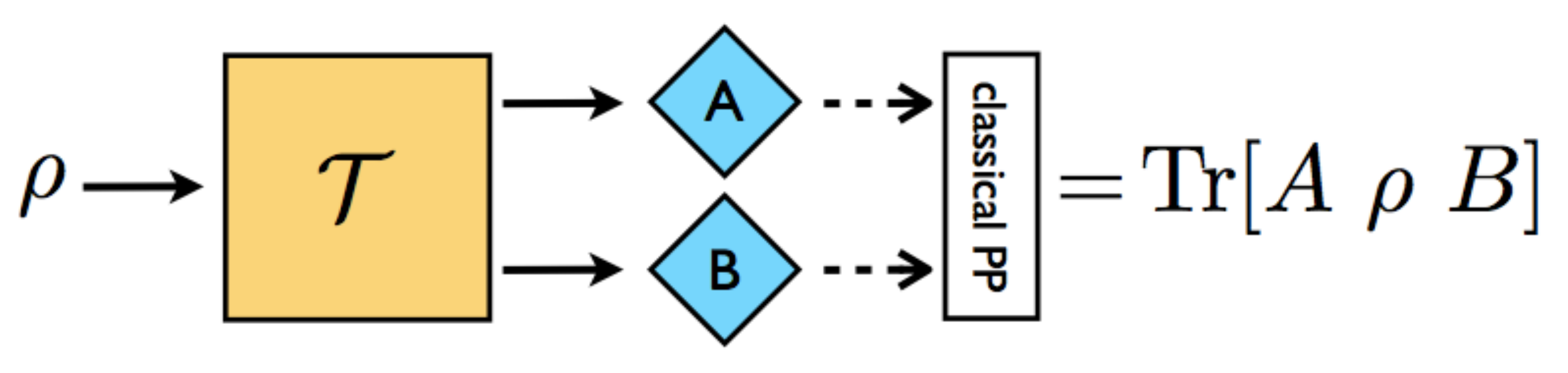}
  \end{center}
  \caption{The \emph{ideal two-point correlator} described
    as a black-box, labeled by $\map{T}$. The quantum
    system, in state $\rho$, is fed into the black-box. The
    output consists of two quantum systems, on which
    independent measurements of observables $A$ and $B$ are
    performed. The data collected are then recombined by a
    purely classical post-processing, resulting in the value
    $\Tr[A\ \rho\ B]$. Notice that both the black-box
    $\map{T}$ and the final classical post-processing are
    independent of $\rho$, $A$ and $B$. In this sense, the
    black-box $\map{T}$ and the post-processing are
    \emph{universal}.}
 \label{fig:black-box}
\end{figure}

 First, the quantum system of interest is fed through a black-box $\map{T}$, what we
 call the ``ideal two-point correlator.'' The black-box in turns produces two output systems, on which the two observables $A$ and $B$ can be independently
 measured, even if they were incompatible. The recorded values are finally post-processed according to a \textit{fixed} post-processing function such that, if the system was initially prepared in state $\rho$, 
 the average result equals $\Tr[A\ \rho\ B]$. Since both the quantum black-box and the classical post-processing are independent of $\rho$, $A$, and $B$, our scheme shows that quantum two-point correlation functions are no less operational than any other expectation value, challenging the common understanding explained above and suggesting, at the same time, new experimental procedures to
 directly measure them.

In the rest of the paper, we will explicitly construct the
black-box and the post-processing function allowing the experimental assessment of any
two-point quantum correlation function of the form
$\Tr[A\ \rho\ B]$, for any state $\rho$ and any pair of observables
$A$ and $B$. Remarkably, both the quantum pre-processing and the
classical post-processing will be independent of $\rho$, $A$
and $B$, thus providing a \emph{universal} strategy. We will
also prove that our strategy is \emph{optimal}, for any
state $\rho$ and any observables $A$ and $B$, in the sense
that it always minimizes the error propagation due to the
final post-processing of data. We will finally present a very simple
probabilistic implementation of our proposal on qubits encoded in the polarization of photons.\medskip

{\em Notation and basic concepts.}---In what follows, we
will only consider quantum systems defined on finite
dimensional Hilbert spaces, denoted by $\sH$ and $\sK$, with
dimensions $d_\sH$ and $d_\sK$, respectively. The set of all
linear operators mapping elements in $\sH$ to elements in
$\sK$ will be denoted by $\linear{\sH,\sK}$, with the
convention that $\linear{\sH}:=\linear{\sH,\sH}$. We will
denote by $\states{\sH}$ the set of all density matrices (or
states), namely all those operators $\rho\in\linear{\sH}$
such that $\rho\ge 0$ and $\Tr[\rho]=1$. The identity matrix
in $\linear{\sH}$ will be denoted by the symbol
$\openone$. In the proofs of our statements, which are
collected in the Supplemental material, we will make use of
well established mathematical results introduced
in~\cite{choi,stine}.\medskip

{\em Formalization.}---We define the \emph{ideal two-point
  quantum correlator} as the linear transformation
$\mT:\linear{\sH}\to\linear{\sH\otimes\sH}$, which acts in
such a way that the following equation,
\begin{align}
  \label{eq:correlator-def}
  \Tr[\map{T}(\rho)\ (A \otimes B) ] := \Tr [ A\ \rho\ B ],
\end{align}
is satisfied for all input states $\rho\in\states{\sH}$ and
all observables $A,B\in\linear{\sH}$. Defining the swap
operator $S\in\linear{\sH\otimes\sH}$ by
$S|\phi,\chi\>=|\chi,\phi\>$ for all $\phi,\chi\in\sH$, the
above equation is equivalent to the following:
\begin{align}
  \label{eq:correlator-action}
  \map{T}(\rho)=S(\openone_\sH\otimes \rho),
\end{align}
for all
$\rho\in\states{\sH}$. Relation~\eqref{eq:correlator-action}
above makes apparent that, on one hand, the map $\mT$ is
linear, but also, on the other, that $\mT$ is not a
\emph{physical} evolution. Such a conclusion is a direct
consequence of the fact that $\mT$ does not preserves
hermiticity, which is a necessary condition for complete
positivity.

However, as we will show in the rest of the paper, even if
the map $\mT$ cannot be realized as a physical evolution, it
can, nonetheless, be given a well motivated operational
interpretation and an experimentally feasible realization
scheme, in terms of \emph{partial expectation values}, a
concept that we will introduce in
Proposition~\ref{prop:physical}.

Before proceeding, we make the following simple
observation. The product of two observables can always be
decomposed as the linear combination of two self-adjoint
operators, namely:
\begin{align*}
  BA = \frac{\{A,B\}}2 -\i \frac{[A,B]}{2\i},
\end{align*}
where $\{A,B\} := AB + BA$ and $[A,B] := AB - BA$ are the
anti-commutator and commutator, respectively, and $\i$
denotes the imaginary unit. By linearity then, any two-point
correlation function can be rewritten as
\begin{align*}
  \Tr[A\ \rho\ B] = \Tr\left[\rho\ \frac{\{A,B\}}2\right]
  -\i \Tr\left[\rho\ \frac{[A,B]}{2\i}\right].
\end{align*}
The above decomposition directly induces an analogous
decomposition of the map $\map{T}$ into its real and
imaginary parts:
\begin{align}
  \label{eq:real-and-im}
  \map{T} = \map{R} - \i \map{I},
\end{align}
where $\map{R}:\linear{\sH}\to\linear{\sH\otimes\sH}$ and $\map{I}:\linear{\sH}\to\linear{\sH\otimes\sH}$ are
defined by
\begin{align}
  \label{eq:real-only}
  \Tr[\map{R}(\rho) \ (A \otimes B)]:= \Tr
  \left[\rho\ \frac{\{A,B\}}{2} \right],
\end{align}
and 
\begin{align}
  \label{eq:im-only}
  \Tr[\map{I}(\rho) \ (A \otimes B)]:=\Tr
  \left[\rho\ \frac{[A,B]}{2\i} \right],
\end{align}
for all $\rho,A,B$.
We notice that, as it was the case for $\mT$, both $\map{R}$ and
$\map{I}$ are linear transformations. Contrarily to $\mT$,
however, they are now both hermiticity-preserving
(HP). Finally, the map $\map{R}$ is easily seen to be also
trace-preserving (TP), while $\Tr[\map{I}(\rho)]=0$, for all
$\rho\in\states{\sH}$.\medskip

{\em Statistical decompositions and partial expectation
  values.}---Suppose that, given a linear HP map
$\map{L}:\linear{\hilb{H}}\to\linear{\hilb{K}}$, one wants
to find a way to experimentally measure the expectation
value $\Tr[\map{L}(\rho)\ A]$, for any input state
$\rho\in\states{\hilb{H}}$ and any observable
$A\in\linear{\hilb{K}}$. The following proposition, proved
in the Supplemental material, provides a way to do so.

\begin{thm}[Statistical Decompositions]\label{prop:decomp}
  Any hermiticity-preserving linear map $\mL:
  \linear{\sH}\to\linear{\sK}$ can be decomposed as
  $\mL=\sum_i\lambda_i\mE_i$ for suitable real coefficients
  $\lambda_i$, where
  $\mE_i:\linear{\sH}\to\linear{\sK}$ are completely
  positive for all $i$, and $\mE:=\sum_i\mE_i$ is
  trace-preserving. [Namely, the collection $\{\mE_i\}_i$
    constitutes a quantum instrument~\cite{instrument}.]
\end{thm}

The content of Proposition~\ref{prop:decomp} is summarized
in Fig.~\ref{fig:stat-sim} below: for any linear HP map
$\map{L}:\linear{\hilb{H}}\to\linear{\hilb{K}}$, there exist
a quantum instrument $\{\map{E}_i\}_i$, with $\map{E}_i
:\linear{\hilb{H}}\to\linear{\hilb{K}}$ CP for all $i$, and
real coefficients $\lambda_i$, such that
\begin{align}
  \label{eq:decomp}
  \Tr[\map{L}(\rho)\ A]=\sum_i\lambda_i
  \Tr[\map{E}_i(\rho)\ A],
\end{align}
for any input state $\rho$ and any observable $A$. Such a
decomposition will be referred to as a \emph{statistical
  decomposition} of the map $\mL$.

\begin{figure}[htb]
  \begin{center}
    \includegraphics[width=\columnwidth]{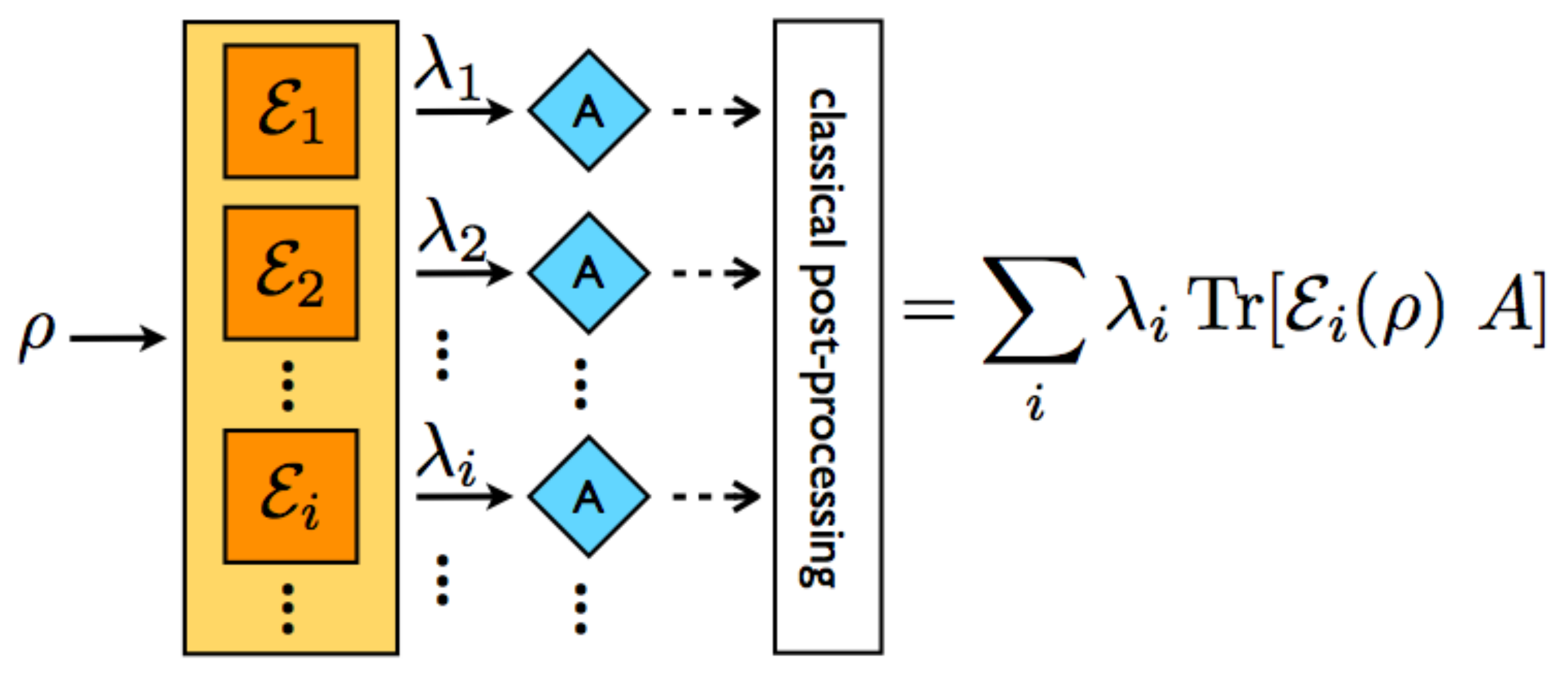}
  \end{center}
  \caption{Statistical decomposition of a non-physical
    transformation: (1)~the initial state $\rho$ goes
    through a quantum instrument, described by a collection
    of CP maps ${\map{E}}_i$; (2)~the outcome $i$, occurring
    with probability $p(i)=\Tr[{\map{E}}_i(\rho)]$, is
    recorded; (3)~the corresponding output state
    $\rho_i={\map{E}}_i(\rho)/p(i)$ is used to evaluate the
    expectation value $\langle A\rangle_i=\Tr[\rho_i\ A]$;
    (4)~all data are finally recombined as
    $\sum_i\lambda_ip(i)\langle A\rangle_i$, for suitable
    real coefficients $\lambda_i$.}
 \label{fig:stat-sim}
\end{figure}

It is important now to notice that, while
Proposition~\ref{prop:decomp} above shows that there always
exists \emph{at least} one statistical decomposition for
every linear HP map, statistical decompositions, as
in~(\ref{eq:decomp}), are in general highly non-unique. In
order to single out an optimal decomposition, an optimality
criterion must be introduced. A natural choice for the
optimality criterion is the statistical
error~\cite{statistics} on the expectation value
$\Tr[\map{L}(\rho) A]$. To define it formally let us rewrite
Eq.~\eqref{eq:decomp} as follows
\begin{align}
  \label{eq:decomp2}
  \Tr[\map{L}(\rho)\ A] & = \sum_i \lambda_i p(i)
  \Tr[\rho_i\ A] \nonumber\\ & = \sum_i \lambda_i p(i)
  \<A\>_i,
\end{align}
where $p(i) := \Tr[ \mE_i (\rho)]$, $\rho_i := p(i)^{-1}
\mE_i(\rho)$, and $\<A\>_i := \Tr[\rho_i\ A]$. Since the
expectation values $\<A\>_i$ all come with their own
statistical error, say $\varepsilon_i$, one has that the
error associated with the linear
combination~\eqref{eq:decomp2} is conservatively evaluated
as $\sum_i| \lambda_i | p(i) \varepsilon_i$. For this
reason, we choose to adopt here the rather conservative
criterion of minimizing $\sum_i|\lambda_i|p(i)$, for all
input states $\rho$.

The following representation theorem (proved in the
Supplemental material) provides an alternative way to interpret
statistical decompositions as \emph{partial
  expectation values} (not to be confused with the
well-established notion of \emph{conditional} expectation
values):
\begin{thm}[Partial Expectation Values]\label{prop:physical}
  For any linear HP map $\mL: \linear{\sH}\to\linear{\sK}$,
  there exists a finite dimensional ancillary quantum system
  $\sK'$, an isometry $V:\sH\to\sK\otimes\sK'$ and an
  observable $Z\in\linear{\sK'}$, such that
  \begin{align}
    \Tr[V\rho V^\dag\ (A\otimes Z)]=\Tr[\mL(\rho)\ A],
  \end{align}
  for all states $\rho\in\states{\sH}$ and all observables
  $A\in\linear{\sK}$. Equivalently,
  \begin{align}
    \mL(\rho)=\Tr_{\sK'}[V\rho V^\dag\ (\openone\otimes Z)],
  \end{align}
  namely, the action of $\mL$ can be written as a ``partial
  expectation value.''
\end{thm}

The idea of partial expectation values is depicted in
Fig.~\ref{fig:real} below.

\begin{figure}[htb]
  \begin{center}
    \includegraphics[width=\columnwidth]{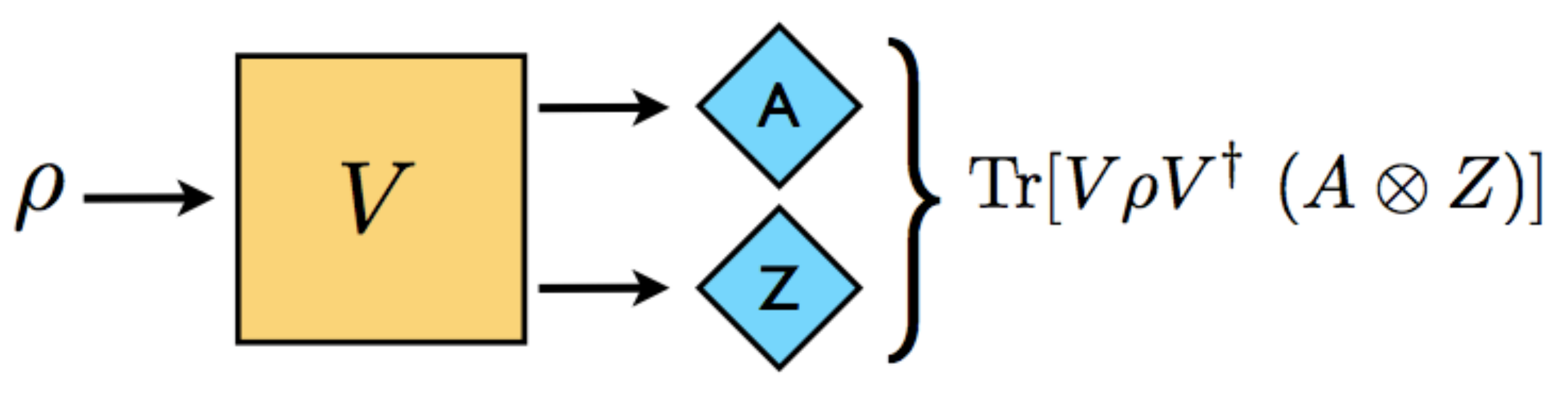}
  \end{center}
  \caption{Statistical decompositions as \emph{partial
      expectation values}: according to
    Proposition~\ref{prop:physical}, $\Tr[V\rho V^\dag\ (A\otimes Z)]=\Tr[\mL(\rho)\ A]$, for
    all input states $\rho$ and all final observables $A$. Notice that the isometry $V$ and
    the ancillary observable $Z$ do not depend neither on
    the input state $\rho$ nor on the final observable $A$,
    but only on the linear HP map $\mL$.}
 \label{fig:real}
\end{figure}

{\em Universal optimal two-point quantum correlator.}---The
proofs of the following Propositions can be found in the
Supplemental material.

\begin{thm}[Universal Optimal Statistical Decomposition of $\map{R}$]
  \label{prop:real-part}
  The map $\map{R}:\linear{\sH}\to\linear{\sH\otimes\sH}$
  representing the real part of the ideal two-point
  correlator $\map{T}$, as in Eqs.~\eqref{eq:real-and-im}
  and~\eqref{eq:real-only}, admits a statistical
  decomposition, which is universally optimal, i.e. optimal
  at once for any input state $\rho\in\states{\sH}$, namely
  \begin{align}
    \label{eq:optimal-dec-real}
    \map{R}=\frac{d_\sH+1}{2}\map{R}_+-\frac{d_\sH-1}{2}\map{R}_-.
  \end{align}
\end{thm}
  In the above equation, $\map{R}_+$ and $\map{R}_-$ are, respectively, the
  symmetric and anti-symmetric $1\to 2$ optimal cloners
  defined, for any $\rho\in\linear{\sH}$, as
  follows~\cite{werner}:
  \begin{align*}
    \map{R}_\pm(\rho):=\frac {2}{d_\sH\pm
      1}P^{\pm}\left(\openone_\sH\otimes\rho\right) P^{\pm},
  \end{align*}
  where $P^+$ and $P^-$ are the projectors on, respectively,
  the symmetric and antisymmetric subspaces of
  $\sH\otimes\sH$, namely, $P^\pm=\frac 12(\openone\pm S)$
  being $S\in\linear{\sH\otimes\sH}$ the swap operator.

\begin{thm}[Universal Optimal Statistical Decomposition of $\map{I}$]
  \label{prop:im-part}
  The map $\map{I}:\linear{\sH}\to\linear{\sH\otimes\sH}$
  representing the imaginary part of the ideal two-point
  correlator $\map{T}$, as in Eqs.~(\ref{eq:real-and-im})
  and~(\ref{eq:im-only}), admits a statistical
  decomposition, which is universally optimal, i.e. optimal
  at once for any input state $\rho\in\states{\sH}$, namely
  \begin{align}
    \label{eq:optimal-dec-im}
    \map{I} = \frac{\sqrt{d_\sH^2-1}}{2} \map{I}_+ -
    \frac{\sqrt{d_\sH^2-1}}{2} \map{I}_-.
  \end{align}
\end{thm}
In the above equation, $\map{I}_+$ and $\map{I}_-$ are
defined, for any $\rho\in\linear{\sH}$, as follows:
\begin{align*}
  \map{I}_\pm(\rho):=\frac{2d_\sH}{(d_\sH^2-1)}
  Q^{\pm}\left(\openone_\sH\otimes\rho\right) Q^{\mp},
\end{align*}
where $Q^+=(Q^-)^\dag:=\frac 12(\openone+zS)$, being
$S\in\linear{\sH\otimes\sH}$ the swap operator and
$z=(-1+\i\sqrt{d_\sH^2-1})/d_\sH$ a complex phase.

{\em Conclusions.}---In this work we provided two point
correlation functions with a new operational
interpretation. We did this by explicitly constructing a
``universal optimal two-point quantum correlator,'' namely,
a measuring apparatus which, independently of $\rho$, $A$,
and $B$, performs an unbiased optimal (in a statistical
sense) estimation of the ideal two-point correlation
function $\Tr[A \ \rho \ B]$. This proves that, despite the
interpretational difficulties due to noncommutativity of $A$
and $B$, two-point correlation functions are as operational
as any other expectation value.

We conclude with a proposal for an experiment
probabilitistically implementing the real part $\map{R}$ of
the universal two-point correlator. Our proposal is depicted
in Fig.~\ref{fig:experiment} in the case of qubits encoded
on photons polarization.  (The case of the imaginary part
$\map{I}$ is more involved: an approximate experimental
implementation, rigorous only in the limit $d\to\infty$,
will be discussed elsewhere, based on results in
Ref.~\cite{CSBDF08}).

\begin{figure}[htb]
  \begin{center}
    \includegraphics[width=\columnwidth]{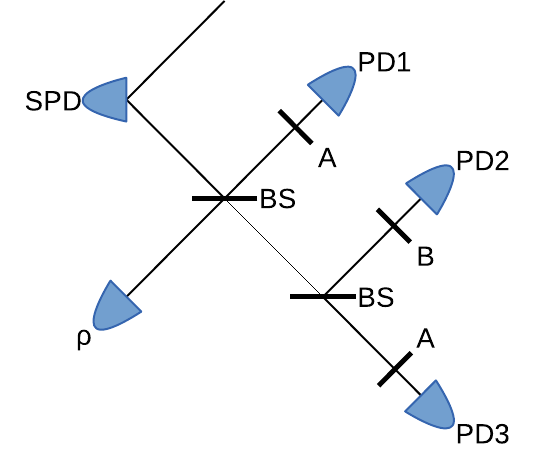}
  \end{center}
  \caption{Experimental proposal for the probabilistic
    implementation of the real part of the universal optimal
    two-point correlator for qubits. Thin lines represent
    optical qubits encoded in the polarization of photons,
    BS is a $50/50$ beamsplitter, SPD is a source of
    maximally entangled photons through spontaneous
    parametric downconversion, $\rho$ is the input state fed
    into the universal optimal two-point correlator, $A$ and
    $B$ are the phase shifters implementing the
    corresponding observables, PD1, PD2, and PD3 are
    photodetectors. A coincidence occurs between PD2 and PD3
    (resp., PD1 and PD2) with probability $3/16$ (resp.,
    $1/16$), in this case optimal universal symmetric
    (resp., antisymmetric) cloning has been performed. Other
    cases are discarded. Averaging over the output
    statistics with the weights given by
    Eq.~\eqref{eq:optimal-dec-real}, one recovers the
    correlation function $\Tr[A \ \rho \ B]$.}
  \label{fig:experiment}
\end{figure}

The system first interacts with a maximally entangled photon
generated by spontaneous parametric downconversion on a
$50/50$ beamsplitter. One of the two output modes is further
splitted by another $50/50$ beamsplitter. Finally, phase
shifters, corresponding to operators $A$ and $B$, are
applied on each output mode, and photodetection is performed
(preceded by polarizing beamsplitters in order to spatially
separate the two polarizations). Despite the present
experimental proposal covers only the case of the real part
(corresponding, as per Eq.~(\ref{eq:real-only}), to the
anti-commutator $\{A,B\}$), it is already enough to provide
new experimental tests of noise-disturbance
relations~\cite{exper-noise-dist} and quantumness
witnesses~\cite{quantumness}.

F. B. was supported by the Program for Improvement of
Research Environment for Young Researchers from Special
Coordination Funds for Promoting Science and Technology
(SCF) commissioned by the Ministry of Education, Culture,
Sports, Science and Technology (MEXT) of Japan. M. D. was
supported by JSPS (Japan Society for the Promotion of
Science) Grant-in-Aid for JSPS Fellows No. 24-0219.
M. O. is supported by the John Templeton Foundations, ID
\#35771, and MIC SCOPE, No. 121806010. This work is
supported by JSPS KAKENHI, No. 21244007.

\onecolumngrid
\newpage

\newpage 

\appendix

\section*{Supplemental material}

In what follows, we will only consider quantum systems
defined on finite dimensional Hilbert spaces, denoted by
$\sH$ and $\sK$, with dimensions $d_\sH$ and $d_\sK$,
respectively. The set of all linear operators mapping
elements in $\sH$ to elements in $\sK$ will be denoted by
$\linear{\sH,\sK}$, with the convention that
$\linear{\sH}:=\linear{\sH,\sH}$. We will denote by
$\states{\sH}$ the set of all density matrices (or states),
namely all those operators $\rho\in\linear{\sH}$ such that
$\rho\ge 0$ and $\Tr[\rho]=1$. The identity matrix in
$\linear{\sH}$ will be denoted by the symbol $\openone$. The
term observable will be used as a synonim of self-adjoint
operator. The identity map on $\linear{\hilb{H}}$ will be
denoted by $\id$. The maximally entangled state $d^{-1}
\sum_{i,j} \ket{i,i} \bra{j,j}$ will be denoted by
$\Phi^+$. The swap operator will be denoted by $S$, namely
$S \in \linear{\hilb{H} \otimes \hilb{H}}$ is the unitary
self-adjoint operator acting as $S \ket{\psi,\phi} =
\ket{\phi,\psi}$, for all $\ket{\psi}, \ket{\phi} \in
\hilb{H}$. The projectors on the symmetric and antisymmetric
subspaces will be denoted by $P_\pm = \frac12 (\openone \pm
S)$, respectively.

Given a linear map $\mL:\linear{\sH}\to\linear{\sK}$, the
so-called \emph{Choi isomorphism}~\cite{choi} provides a way
to associate $\mL$ with an operator
$\choi{\mL}\in\linear{\sK\otimes\sH}$, whose matrix, in the
standard representation given by the computational basis
$\{|k\>\}$, is defined as
\begin{align}
  \label{eq:choi-def}
  \choi{\mL}:=d(\mL\otimes\id)(|\Phi^+\>\<\Phi^+|),
\end{align}
being $|\Phi^+\>$ the standard maximally entangled state
introduced above. The inverse correspondence works as
follows: given an operator $J\in\linear{\sK\otimes\sH}$, the
Choi isomorphism constructs a linear map
$\invchoi{J}:\linear{\sH}\to\linear{\sK}$ defined, for all
$M\in\linear{\sH}$, by
\begin{align}
  \label{eq:choi-inv}
  \invchoi{J}(M) := \Tr_\sH\left[J\ (\openone_{\sK}\otimes
    M^T)\right],
\end{align}
where the exponent $T$ denotes the transposition with
respect to the computational basis $\{|k\>\}$. The
importance of the Choi isomorphism lies in the following three properties:
\begin{enumerate}
\item linearity,
  i.e. $\choi{a\mL+b\mL'}=a\choi{\mL}+b\choi{\mL'}$ and
  $\invchoi{aJ+bJ'}=a\invchoi{J}+b\invchoi{J'}$;
\item bijectivity, i.e. $\invchoi{\choi{\mL}}=\mL$ and
  $\choi{\invchoi{J}}=J$;
\item finally, and more importantly, the map $\mL$ is
  completely positive (CP) if and only if the corresponding
  operator $\choi{\mL}$ is non-negative.
\end{enumerate}
Other properties, which follow easily from the definition,
are the following:
\begin{enumerate}
\setcounter{enumi}{3}
\item the map $\mL$ is hermiticity-preserving (HP), if and
  only if the corresponding operator $\choi{\mL}$ is
  hermitian;
\item the map $\mL$ is trace-preserving (TP), if and only if
  the corresponding operator $\choi{\mL}$ satisfies the
  normalization condition
  $\Tr_{\sK}[\choi{\mL}]=\openone_\sH$.
\end{enumerate}

We now prove Proposition~\ref{prop:decomp}, that we restate here
for convenience.
\begin{thm}[Statistical Decompositions]
  Any hermiticity-preserving linear map $\mL:
  \linear{\sH}\to\linear{\sK}$ can be decomposed as
  $\mL=\sum_i\lambda_i\mE_i$ for suitable coefficients
  $\lambda_i\in\mathbb{R}$, where
  $\mE_i:\linear{\sH}\to\linear{\sK}$ are completely
  positive for all $i$, and $\mE:=\sum_i\mE_i$ is
  trace-preserving.
\end{thm}

\begin{proof}
  We already saw that the operator
  $\choi{\mL}\in\linear{\sK\otimes\sH}$ is hermitian,
  whenever the map $\mL$ is HP. We can therefore diagonalize
  $\choi{\mL}$ as $\choi{\mL}=\sum_i\mu_i\Pi_i$, with
  $\mu_i\in\mathbb{R}$ and $\Pi_i$ orthogonal projectors
  such that
  $\sum_i\Pi_i=\openone_{\sK}\otimes\openone_{\sH}$. The
  statement is recovered simply by normalizing by $d_\sK$,
  namely, $\lambda_i:=d_\sK\mu_i$, and
  $\mE_i:=\invchoi{d_\sK^{-1}\Pi_i}$.
\end{proof}

The following Lemma provides a lower bound on the
statistical error $\sum_i|\lambda_i|p(i)$, with $p(i) :=
\Tr[ \mE_i (\rho)]$, associated to statistical decomposition
$\mL=\sum_i\lambda_i\mE_i$.
\begin{lmm}\label{prop:optimal-cond}
  Given a HP linear map $\mL: \linear{\sH}\to\linear{\sK}$,
  for any statistical decomposition
  $\mL=\sum_i\lambda_i\mE_i$ and any input state
  $\rho\in\states{\sH}$,
  \begin{equation}\label{eq:optimal-cond}
    \sum_i|\lambda_i|p(i) \ge \min_{\sigma\in\states{\sH}}
    \Tr \left[ |\choi{\mL}| \ (\openone_\sK\otimes\sigma)
      \right],
  \end{equation}
  where $p(i):=\Tr[\mE_i(\rho)]$.
\end{lmm}

\begin{proof}
  For any statistical decomposition
  $\mL=\sum_i\lambda_i\mE_i$, the linearity of the Choi
  isomorphism implies that
  $\choi{\mL}=\sum_i\lambda_i\choi{\mE_i}$. This implies that
  $|\choi{\mL}|\le\sum_i|\lambda_i|\choi{\mE_i}$, which in turn implies $\Tr\left[|\choi{\mL}|\ (\openone_\sK\otimes\sigma)\right]\le\sum_i|\lambda_i|\Tr\left[\choi{\mE_i}\ (\openone_\sK\otimes\sigma)\right]$ for all $\sigma\ge0$. The statement is recovered by minimizing over $\sigma$ the left-hand side.
\end{proof}

We now prove Proposition~\ref{prop:physical}, that we
restate for convenience.
\begin{thm}[Partial Expectation Values]
  For any linear HP map $\mL: \linear{\sH}\to\linear{\sK}$,
  there exists a finite dimensional ancillary quantum system
  $\sK'$, an isometry $V:\sH\to\sK\otimes\sK'$ and an
  observable $Z\in\linear{\sK'}$, such that
\begin{equation}
  \Tr[V\rho V^\dag\ (A\otimes Z)]=\Tr[\mL(\rho)\ A],
\end{equation}
for all states $\rho\in\states{\sH}$ and all observables
$A\in\linear{\sK}$. Equivalently,
\begin{equation}
  \mL(\rho)=\Tr_{\sK'}[V\rho V^\dag\ (\openone\otimes Z)],
\end{equation}
namely, the action of $\mL$ can be written as a partial
expectation value.
\end{thm}

\begin{proof}
  Let $\mL(\rho)=\sum_i\lambda_i\mE_i(\rho)$ be a
  statistical decomposition of $\mL$. Then, following
  Stinespring's representation theorem~\cite{stine}, there
  exist $\sK'$ ancillary Hilbert space,
  $V:\sH\to\sK\otimes\sK'$ isometry, and $\{P^i\}_i$ POVM on
  $\sK'$ such that
\begin{equation*}
  \mE_i(\rho)=\Tr_{\sK'}[V\rho
    V^\dag\ (\openone_{\sK}\otimes P^i_{\sK'})].
\end{equation*}
The statement is recovered by setting
$Z:=\sum_i\lambda_iP^i$.
\end{proof}

According to Eqs.~\eqref{eq:real-only}
and~\eqref{eq:im-only}, the real part $\map{R}$ and the
imaginary part $\map{I}$ of the ideal two-point correlator
are defined as
\begin{align}
  \label{eq:timerdef}
  \Tr[(A \otimes B) \ \map{R}(\rho)] := \Tr \left[
    \frac{\{A,B\}}{2}\ \rho \right],\\
  \label{eq:timeirdef}
  \Tr[(A \otimes B) \ \map{I}(\rho)] := \Tr \left[
    \frac{[A,B]}{2 \i}\ \rho \right],
\end{align}
for any observables $A, B$ and any state $\rho$. Their
action can be written as
\begin{align}
  \label{eq:timeract}
  \map{R}(\rho) = \frac{(\openone \otimes \rho) S + S
    (\openone \otimes \rho)}{2},\\
  \label{eq:timeiact}
  \map{I}(\rho) = \frac{(\openone \otimes \rho) S - S
    (\openone \otimes \rho)}{2\i},
\end{align}
where $S$ is the swap operator, and their Choi operators are
given by
\begin{align}
  \label{eq:timerchoi}
  \choi{\map{R}} = \frac{d}2 \left[ (\openone_1 \otimes
    \Phi^+_{2,3}) (S_{1,2} \otimes \openone_3) + (S_{1,2}
    \otimes \openone_3) (\openone_1 \otimes \Phi^+_{2,3})
    \right],\\
      \label{eq:timeichoi}
    \choi{\map{I}} = \frac{d}{2\i} \left[ (\openone_1 \otimes
      \Phi^+_{2,3}) (S_{1,2} \otimes \openone_3) - (S_{1,2}
      \otimes \openone_3) (\openone_1 \otimes \Phi^+_{2,3})
      \right].
\end{align}

Let us introduce maps $\map{R}_\pm$ and $\map{I}_\pm$ by
giving their Choi operators
\begin{align}
  \label{eq:clonerchoi}
  \choi{\map{R}_\pm} := \frac {2d}{d\pm 1} (P^\pm_{1,2}
  \otimes \openone_3) (\openone_1 \otimes \Phi^+_{2,3})
  (P^\pm_{1,2} \otimes \openone_3),\\
  \label{eq:rootswapchoi}
  \choi{\map{I}_\pm} := \frac{2d^2} {d^2 - 1} (Q^\pm_{1,2}
  \otimes \openone_3) (\openone_1 \otimes \Phi^+_{2,3})
  (Q^\mp_{1,2} \otimes \openone_3),
\end{align}
where $P^\pm := \frac12 (\openone \pm S)$ are the projectors
on the symmetric and antisymmetric subspace, respectively,
and $Q^+ = (Q^-)^\dag := \frac12 (\openone + z S)$, being
$S$ the swap operator and $z=(-1+\i\sqrt{d^2-1})/d$ a
complex phase. Maps $\map{R}_\pm$ and $\map{I}_\pm$ are
completely positive and trace preserving. We notice that map
$\map{R_+}$ is the universal optimal quantum
cloning~\cite{werner}.

We can now prove Propositions~\ref{prop:real-part}
and~\ref{prop:im-part}, that we restate for convenience.
\begin{thm}
  The map $\map{R}$ admits a statistical decomposition which
  is universally optimal, i.e. optimal at once for any input
  state $\rho\in\states{\sH}$, namely
  \begin{align}
    \label{eq:optimal-dec-re2}
    \map{R} = \frac{d_{\sH}+1}{2} \map{R}_+ -
    \frac{d_\sH-1}{2} \map{R}_-.
  \end{align}
\end{thm}

\begin{proof}
  The fact that Eq.~\eqref{eq:optimal-dec-re2} is a
  statistical decomposition follows by direct inspection.

  For optimality, notice that the right-hand side of
  Eq.~\eqref{eq:optimal-cond} can be explicitly computed as
  \begin{align*}
    & \min_{\sigma\in\states{\sH}} \Tr \left[
      |\choi{\map{R}}|
      \ (\openone_{\sH\otimes\sH}\otimes\sigma) \right] \\ =
    & \frac{d_\sH+1}2 \min_{\sigma\in\states{\sH}} \Tr
    \left[ |\choi{\map{R}_+}|
      \ (\openone_{\sH\otimes\sH}\otimes\sigma) \right] +
    \frac{d_\sH-1}2 \min_{\sigma\in\states{\sH}} \Tr \left[
      |\choi{\map{R}_-}|
      \ (\openone_{\sH\otimes\sH}\otimes\sigma) \right] \\ =
    & d_\sH,
  \end{align*}
  where first inequality follows from the orthogonality and
  positive semidefiniteness of $\choi{\map{R}_\pm}$ and the
  second equality follows from the fact that $\map{R}_\pm$ are trace-preserving, i.e.
  $\Tr_\sK[\choi{\map{R}_\pm}] = \openone_\sH$. The
  decomposition~(\ref{eq:optimal-dec-real}), once rewritten
  in the form of Proposition~\ref{prop:decomp}, becomes
  \begin{equation*}
    \map{R} =
    \lambda_+\frac{\map{R}_+}{2}-\lambda_-\frac{\map{R}_-}{2},
  \end{equation*}
  where $\lambda_\pm:=(d_\sH \pm 1)$, due to the fact that
  both $\map{R}_+$ and $\map{R}_-$ are CPTP, and, therefore,
  $\frac{\map{R}_+}{2}$ and $\frac{\map{R}_-}{2}$ constitute
  a legitimate quantum instrument. By direct evaluation, the
  left-hand side of Eq.~(\ref{eq:optimal-cond}) is equal to
  $(d_\sH+1)/2 + (d_\sH-1)/2 = d_\sH$ for any input state
  $\rho$, because $p(+)=p(-)=1/2$ for any state
  $\rho$. Therefore, the optimality holds \emph{for any}
  input state $\rho$.
\end{proof}

\begin{thm}
  The map $\map{I}$ admits a statistical decomposition,
  which is universally optimal, i.e. optimal at once for any
  input state $\rho\in\states{\sH}$, namely
  \begin{align}
    \label{eq:optimal-dec-im2}
    \map{I} = \frac{\sqrt{d_\sH^2-1}}{2} \map{I}_+ -
    \frac{\sqrt{d_\sH^2-1}}{2} \map{I}_-,
  \end{align}
\end{thm}

\begin{proof}
  The fact that Eq.~\eqref{eq:optimal-dec-im2} is a
  statistical decomposition follows by direct inspection.

  For optimality, notice that the right-hand side
  of~(\ref{eq:optimal-cond}) can be explicitly computed as
  \begin{align*}
    & \min_{\sigma\in\states{\sH}} \Tr \left[
      |\choi{\map{I}}|
      \ (\openone_{\sH\otimes\sH}\otimes\sigma) \right] \\ =
    & \frac{\sqrt{d_\sH^2-1}}2 \min_{\sigma\in\states{\sH}}
    \Tr \left[ |\choi{\map{I}_+}|
      \ (\openone_{\sH\otimes\sH}\otimes\sigma) \right] +
    \frac{\sqrt{d_\sH^2-1}}2 \min_{\sigma\in\states{\sH}}
    \Tr \left[ |\choi{\map{I}_-}|
      \ (\openone_{\sH\otimes\sH}\otimes\sigma) \right] \\ =
    & \sqrt{d_\sH^2-1},
  \end{align*}
  where first inequality follows from the orthogonality and
  positive semidefiniteness of $\choi{\map{I}_\pm}$ and the
  second equality follows from the fact that $\map{I}_\pm$ are trace-preserving, i.e.
  $\Tr_\sK[\choi{\map{I}_\pm}] = \openone_\sH$. The proof of
  orthogonality between $\choi{\map{I}_\pm}$ is lengthy but
  not difficult, the details will be spelled out in a
  forthcoming paper by the present authors. The
  decomposition~(\ref{eq:optimal-dec-im}), once rewritten in
  the form of Proposition~\ref{prop:decomp}, becomes
  \begin{equation*}
    \map{I} =
    \lambda\frac{\map{I}_+}{2}-\lambda\frac{\map{I}_-}{2},
  \end{equation*}
  where $\lambda:=\sqrt{d_\sH^2-1}$. Arguments, analogous to
  those used in the proof of
  Proposition~\ref{prop:real-part}, show that the optimality
  holds \emph{for any} input state $\rho$.
\end{proof}

\end{document}